\documentclass[conference]{IEEEtran}
\usepackage{amsmath,amssymb,amsfonts}
\usepackage{graphicx}
\usepackage{booktabs}
\usepackage{cite}
\usepackage{siunitx}
\usepackage{array}
\usepackage{amsthm}
\usepackage{bm}
\usepackage{algorithm}
\usepackage{algorithmic}

\newtheorem{theorem}{Theorem}

\usepackage{fancyhdr}
\fancypagestyle{firststyle}{\fancyhf{}
		\fancyhead[L]{\small M. Ying, G. Qian, X. Wang, X. Liu, I. S. Gupte, P. Ma, D. Shakya, and T. S. Rappaport, ``SPARC: Sparse Path-Aware Residual Calibrator for Wireless Ray Tracing at Upper Mid-Band," in \textit{2026 IEEE Global Communications Conference (GLOBECOM)}, Macao, China, Dec. 2026, pp. 1--6.}
	}

\begin{document}
\setlength{\columnsep}{0.22in}
\addtolength{\topmargin}{5pt}
\addtolength{\textheight}{-1pt}
\setlength{\textfloatsep}{6pt plus 2pt minus 2pt}
\setlength{\intextsep}{6pt plus 2pt minus 2pt}

\title{SPARC: Sparse Path-Aware Residual Calibrator for Wireless Ray Tracing at Upper Mid-Band}

\author{\IEEEauthorblockN{Mingjun Ying\IEEEauthorrefmark{1},  Guanyue Qian, Xinquan Wang, Xingchen Liu, Ishaan S. Gupte, Peijie Ma,\\ Dipankar Shakya, and Theodore S. Rappaport\IEEEauthorrefmark{2}}
\IEEEauthorblockA{NYU WIRELESS, New York University, Brooklyn, NY, USA}
\IEEEauthorblockA{\IEEEauthorrefmark{1}yingmingjun@nyu.edu, \IEEEauthorrefmark{2}tsr@nyu.edu}}

\maketitle
\thispagestyle{firststyle}
\bstctlcite{BSTcontrol}
\begin{abstract}
Accurate site-specific ray tracing (RT) is essential for upper mid-band network planning, yet raw RT can produce per-path multipath component (MPC) power errors on the order of 19--24~dB in cluttered indoor environments. A fixed-geometry material-sensitivity bound shows that a 30\% relative-permittivity perturbation changes each surface interaction by at most 6.28~dB across the considered indoor materials. However, even MPCs with only one surface interaction exhibit a 19.2~dB mean RT--measurement bias, suggesting that missing clutter, displaced surfaces, and simplified 3D geometry dominate the per-path RT error. We propose SPARC (Sparse Path-Aware Residual Calibrator), a lightweight per-path calibration method that learns a sparse linear residual model from one completed RT simulation. SPARC uses standard RT features selected per fold by nested cross-validation, with ridge regularization and power-gated path matching; four features recur in both environments. Using measured indoor factory (InF) and indoor hotspot (InH) datasets at 6.75 and 16.95~GHz, SPARC reduces per-path power RMSE from 18.74 to 4.74~dB in InF and from 23.12 to 5.39~dB in InH. A jointly trained InF+InH model achieves 5.73~dB RMSE. When all links from one transmitter location are held out for testing, SPARC achieves 4.99~dB RMSE in InF and 5.85~dB RMSE in InH. SPARC therefore provides a practical post-processing calibration layer for site-specific per-path power prediction without ray-tracer modification or additional RT runs.
\end{abstract}


\section{Introduction}
\label{sec:intro}

Site-specific RT has long been used to predict path loss, delay spread, and dominant propagation paths in indoor and outdoor environments~\cite{mcknown1991ray,seidel1994site,panjwani1996coverage,yun2015ray}.
Modern RT tools such as Wireless InSite~\cite{remcom_winsite}, Sionna~RT~\cite{hoydis2023sionna}, and NYURay~\cite{kanhere2024nyuray} have extended site-specific channel modeling to upper mid-band and mmWave frequencies for 5G and 6G systems~\cite{rappaport2013mmwave,ying2026npj,3gpp_inf}.
Related NYU work spans indoor 72~GHz measurements, 28~GHz multi-beam propagation, 142~GHz urban microcells, and multi-frequency statistical modeling~\cite{nie2013indoor,sun2014multibeam,xing2021subthz,samimi2015statistical}; earlier site-specific systems addressed building-database construction and reception-surface ray tracing~\cite{rappaport2004BDM,rappaport2004RTsurface}.
For practical wireless planning, however, matching measured channel behavior matters more than geometric detail alone.
Raw RT output often shows large power error because electromagnetic material parameters, clutter placement, and geometry registration jointly affect both the attenuation of traced paths and the candidate paths retained by the RT solver.
RT calibration is therefore necessary before site-specific prediction can be used reliably in cluttered indoor environments.


Most prior RT calibration methods tune material parameters at link level: simulated annealing~\cite{jemai2009calibration}, coordinate descent~\cite{kanhere2024nyuray}, and differentiable RT~\cite{hoydis2024learning,ruah2024calibrating} require repeated RT evaluations. Material tuning modifies Fresnel loss only along existing paths, whereas clutter omission and geometry simplification change path topology. Per-path attenuation based on scatterer size~\cite{haniz2025multipath} requires object-level geometry rarely available in standard RT outputs. Automated 3D reconstruction with material assignment~\cite{ying2026horama} and location calibration~\cite{ying2025pdp_location,ying2026npj} address scene-model errors. Wireless digital twins support indoor navigation and localization under virtual--physical mismatch~\cite{li2025digitaltwin,lei2026locus}, but neither study calibrates individual RT-path power. Measurement-driven statistical channel modeling likewise emphasizes physically consistent datasets and standardized representations~\cite{jariwala2026nyusim}; such models complement site-specific RT but do not retain individual traced-path identities. Ensemble ML methods, including gradient boosting~\cite{chen2016xgboost}, are common baselines~\cite{vasudevan2024mlsurvey,benameur2025xgboost} but may degrade on unseen geometries~\cite{celades2026shadow}. Table~\ref{tab:related_work} categorizes existing approaches by calibration domain, resolution, and operational cost.

The central question is whether path-level features from a completed RT simulation can correct per-path power error that material calibration probably cannot. A fixed-geometry material-correctable bound (Theorem~\ref{thm:material_bound}) shows that the observed 19--23~dB per-bounce bias cannot be explained by material-parameter error alone; the remaining gap originates from 3D scene inaccuracies that cause RT to trace incorrect paths. Sparse Path-Aware Residual Calibrator (SPARC) exploits path-level features for a sparse linear correction from one completed simulation. Validation against InF~\cite{ying2025icc_infactory} and InH~\cite{shakya2024ojcoms} measurements at 6.75 and 16.95~GHz shows that four common RT features improve per-path prediction in both environments. 
The main contributions of this paper are:

\begin{table}[!t]
\centering
\caption{RT calibration approaches by optimization domain, resolution level, and need for additional RT evaluations.}
\label{tab:related_work}
\small
\setlength{\tabcolsep}{3pt}
\begin{tabular}{@{}lccc@{}}
\toprule
Method & Domain & Level & Extra RT \\
\midrule
Simulated annealing~\cite{jemai2009calibration} & Material & Link & Yes \\
Coord.\ descent~\cite{kanhere2024nyuray} & Material & Link & Yes \\
Diff.\ RT~\cite{hoydis2024learning} & Material & Link & Yes \\
Scatterer atten.~\cite{haniz2025multipath} & Material & Path & No \\
NeRF / 3DGS~\cite{mildenhall2021nerf,kerbl2023gaussian} & Neural & Link & N/A \\
\textbf{SPARC (proposed)} & Metadata & \textbf{Path} & \textbf{No} \\
\bottomrule
\end{tabular}
\par\vspace{3pt}{\footnotesize\raggedright\emph{Level}: Link = whole-link power; Path = per matched MPC. \emph{Extra RT}: additional RT runs beyond initial simulation.\par}
\vspace{-3pt}
\end{table}

\begin{itemize}
\item \textbf{Per-path power calibration from a single RT simulation.} SPARC learns a sparse linear correction from a selected feature subset and requires no additional RT evaluations at deployment. Validated against 1~GHz-bandwidth field measurements~\cite{ying2025icc_infactory,shakya2024ojcoms} at 6.75 and 16.95~GHz, SPARC reduces per-path power RMSE from 18.74 to 4.74~dB in InF and from 23.12 to 5.39~dB in InH.
\item \textbf{Material-sensitivity bound proving geometry dominates.} Theorem~\ref{thm:material_bound} shows that $\pm$30\% relative-permittivity perturbations shift each surface interaction by at most 1.91~dB for concrete and 1.68~dB for glass, rising to 3.74~dB for plywood and 6.28~dB for wood. Measured per-bounce bias exceeds 19~dB, proving that material calibration cannot close the gap.
\item \textbf{Generalization across indoor environments.} Training a single SPARC model on pooled InF and InH measurement data achieves 5.73~dB per-path power RMSE. Holding out entire TX locations for testing yields 4.99~dB (InF) and 5.85~dB (InH) RMSE, confirming that the learned correction transfers to unseen spatial locations.
\end{itemize}

\section{System Model and Problem Formulation}
\label{sec:system}

\subsection{Wideband Channel and Measurement Model}

Consider a wideband single-input single-output link between a transmitter (TX) and a receiver (RX).
The omnidirectional channel impulse response is
\begin{equation}
h(\tau) = \sum_{\ell=1}^{L} a_{\ell}\,\delta(\tau-\tau_{\ell}),
\label{eq:cir}
\end{equation}
where $a_{\ell}$ and $\tau_{\ell}$ denote the complex amplitude and propagation delay of path $\ell$.
Both campaigns use a 1~GHz bandwidth sliding-correlator sounder at 6.75 and 16.95~GHz; full sounder specifications are given in~\cite{ying2025icc_infactory,shakya2024ojcoms}.
The stored PDP is a band-limited observation of~\eqref{eq:cir} on a discrete delay grid.
NYURay~\cite{kanhere2024nyuray,ying2026npj} returns idealized path tuples $(a_{\ell}^{\mathrm{RT}},\tau_{\ell}^{\mathrm{RT}})$ together with interaction count, material labels, and angular metadata.
RT paths are convolved with a $\mathrm{sinc}(Bt)$ kernel ($B=1$~GHz) and summed on the measurement delay grid to match the sounder bandwidth.
The calibration target is the integrated energy of each dominant matched MPC in dBm on the resulting band-limited PDP.
{RT and the physical channel produce different path sets because the 3D model may omit clutter or simplify surfaces, so per-path calibration requires pairing each RT path with the corresponding measured peak before power error can be computed.}

\subsection{Indoor Measurement and Matched-Path Problem}

The InF campaign~\cite{ying2025icc_infactory} covers three TX and twelve RX locations in a 36$\times$22~m$^{2}$ open-plan factory with workstations, metallic equipment, glass partitions, and concrete pillars (T--R separation 9--37~m; floor plan and TX/RX layout in~\cite{ying2025icc_infactory}).
The InH campaign~\cite{shakya2024ojcoms} covers four TX and twenty RX locations in a 109$\times$44~m$^{2}$ corridor-connected multi-room office building with concrete walls, metal furnishings, and wooden doors (T--R separation 11--97~m; floor plan and TX/RX layout in~\cite{shakya2024ojcoms}). Measurements capture MPC absolute time with sub-nanosecond accuracy using PTP-based synchronization, supporting accurate RT path identification~\cite{shakya2023sub}.

RT paths are paired one-to-one with measured dominant peaks by the Hungarian algorithm~\cite{kuhn1955hungarian} under a 10~ns delay tolerance.
A power gate $\Delta_{\max}$ rejects training candidate matches whose absolute power difference exceeds $\Delta_{\max}$ ($\Delta_{\max}=30$~dB for InF, 38~dB for InH), removing delay-coincident pairs that arrive via different physical routes due to 3D model inaccuracies such as missing glass partitions; each retained pair is a matched MPC.
{The wider InH gate reflects the larger T--R separation range (11--97~m versus 9--37~m in InF), which increases the spread of received MPC powers and requires a wider acceptance window.}
{The gate is a training-data cleaning step; at deployment SPARC is applied to every RT path, raising per-path RMSE by 0.4/1.0~dB (InF/InH) as an upper bound.}
For matched path $\ell$, let $P_{\ell}^{\mathrm{RT}}$ denote the RT power in dBm and let $P_{\ell}^{\mathrm{meas}}$ denote the measured power.
The matched-path error is
\begin{equation}
e_{\ell} \triangleq P_{\ell}^{\mathrm{RT}} - P_{\ell}^{\mathrm{meas}},
\label{eq:error}
\end{equation}
decomposed as $e_{\ell} = e_{\ell}^{\mathrm{mat}} + e_{\ell}^{\mathrm{geo}} + \nu_{\ell}$.
The decomposition is additive in dB because each interaction contributes a multiplicative Fresnel factor in linear power (cf.~\eqref{eq:path_model}).
Here $e_{\ell}^{\mathrm{mat}}$ denotes the contribution of imperfect material parameters on an otherwise fixed traced path, $e_{\ell}^{\mathrm{geo}}$ denotes the contribution of path-topology mismatch such as clutter omission or scene abstraction, and $\nu_{\ell}$ collects residual mismatch from peak extraction and unmodeled effects.
{Eq.~\eqref{eq:error} is a conceptual decomposition: $e_{\ell}^{\mathrm{mat}}$ and $e_{\ell}^{\mathrm{geo}}$ are not separately identifiable, and the exclusion argument in Sec.~\ref{sec:method_bound} only requires $|e_{\ell}^{\mathrm{mat}}|\leq n_{\ell}\eta_{\max}$ on the RT-traced path, which holds whether or not the path is physically realized (e.g., when a missing 3D-model object causes a phantom direct trace).}
Let $\mathbf{x}_{\ell}\in\mathbb{R}^{19}$ denote the RT feature vector extracted from matched path $\ell$ and $q_{\ell}$ the corresponding TX--RX--frequency group.
Let $N$ denote the total number of power-gated matched MPCs pooled across all TX--RX links and both carrier frequencies.

Material-domain calibration perturbs the nominal material vector $\bm{\vartheta}$ and solves
\begin{equation}
\min_{\Delta\bm{\vartheta}\in\mathcal{V}}\; \frac{1}{N}\sum_{\ell=1}^{N}
\left(P_{\ell}^{\mathrm{RT}}(\bm{\vartheta}+\Delta\bm{\vartheta})-P_{\ell}^{\mathrm{meas}}\right)^2.
\label{eq:material_problem}
\end{equation}
Residual-space calibration instead seeks a correction function $g(\mathbf{x}_{\ell})$ such that
$\widehat{P}_{\ell}^{\mathrm{cal}} = P_{\ell}^{\mathrm{RT}} - g(\mathbf{x}_{\ell})$ and
\begin{equation}
\min_{g \in \mathcal{G}}\; \frac{1}{N}\sum_{\ell=1}^{N}\left(e_{\ell}-g(\mathbf{x}_{\ell})\right)^2.
\label{eq:problem}
\end{equation}
Problem~\eqref{eq:material_problem} can only reduce $e_{\ell}^{\mathrm{mat}}$ along paths whose topology is unchanged by the perturbation, whereas~\eqref{eq:problem} can absorb any predictable error component, including scene-geometry errors, encoded by $\mathbf{x}_{\ell}$.

\section{Proposed Calibration Method}
\label{sec:method}

SPARC takes matched MPC pairs from one completed RT simulation as input and outputs calibrated per-MPC power via sparse ridge regression; Algorithm~\ref{alg:sparse_calibration} details the procedure.
Sec.~\ref{sec:method_bound} first establishes a material-sensitivity bound proving that a residual-space approach is necessary.

\subsection{Material-Correctable Bound Under Fixed Geometry}
\label{sec:method_bound}

For one matched path with fixed interaction sequence, fixed incidence angles, and fixed path length, the ray-traced power in dB can be written as
\begin{equation}
P_{\ell}^{\mathrm{RT}}(\bm{\vartheta}) = C_{\ell}^{\mathrm{geom}} + \sum_{k=1}^{n_{\ell}} \phi_{m_k}(\bm{\vartheta}_{m_k},\gamma_k),
\label{eq:path_model}
\end{equation}
where $n_{\ell}$ is the number of interactions on path $\ell$, $C_{\ell}^{\mathrm{geom}}$ collects all geometry-fixed terms such as spreading loss and antenna factors, $\bm{\vartheta}_{m}=[\varepsilon_{r,m},\sigma_m]^\top$ contains the relative permittivity and conductivity of material $m$, $m_k$ is the material at interaction $k$, and $\gamma_k$ collects the local geometric arguments such as incidence angle, polarization state, and interaction type.
One explicit planar-interface realization is
\begin{equation}
\begin{aligned}
\phi_m(\bm{\vartheta}_m,\gamma) &= 20\log_{10}\left|F_{p,t}(\widetilde{\varepsilon}_{r,m},\theta_i)\right|,\\
\widetilde{\varepsilon}_{r,m} &= \varepsilon_{r,m}-j\frac{\sigma_m}{\omega\varepsilon_0},
\end{aligned}
\label{eq:fresnel_phi}
\end{equation}
where $F_{p,t}$ denotes the Fresnel reflection or transmission coefficient selected by polarization $p$, interaction type $t$, and incidence angle $\theta_i$ measured from the surface normal.

Perturbing $\varepsilon_r$ within $\pm 30\%$ of nominal ITU values shifts each interaction by under 2~dB for structural materials, 3.74~dB for plywood, and 6.28~dB for wood, so an $n$-bounce path changes by at most $n$ times the per-interaction ceiling.
\begin{theorem}[Material-correctable bound]
\label{thm:material_bound}
Assume each interaction function $\phi_{m_k}(\cdot,\gamma_k)$ is twice continuously differentiable on a compact convex admissible set $\mathcal{V}_{m_k}\subset\mathbb{R}^{2}$, and the perturbation satisfies $\lVert \Delta \bm{\vartheta}_m \rVert_2 \leq \rho_m$ for every material~$m$.
Define the first-order and second-order sensitivity envelopes:
\begin{align}
S_m(\gamma) &\triangleq \sup_{\mathbf{u}\in\mathcal{V}_m}\lVert\nabla_{\mathbf{u}}\phi_m(\mathbf{u},\gamma)\rVert_2, \label{eq:S_env}\\
H_m(\gamma) &\triangleq \sup_{\mathbf{u}\in\mathcal{V}_m}\lVert\nabla_{\mathbf{u}}^2\phi_m(\mathbf{u},\gamma)\rVert_2. \label{eq:H_env}
\end{align}
Let $\Delta P_{\ell}^{\mathrm{mat}} \triangleq P_{\ell}^{\mathrm{RT}}(\bm{\vartheta}+\Delta \bm{\vartheta}) - P_{\ell}^{\mathrm{RT}}(\bm{\vartheta})$ denote the path-power change due to a material perturbation on a fixed traced path.
Then every matched path $\ell$ obeys
\begin{equation}
\left|\Delta P_{\ell}^{\mathrm{mat}}\right|
\leq \sum_{k=1}^{n_{\ell}}
\Bigl(
S_{m_k}(\gamma_k)\,\rho_{m_k}
+ \tfrac{1}{2}\,H_{m_k}(\gamma_k)\,\rho_{m_k}^{2}
\Bigr).
\label{eq:bound}
\end{equation}
\end{theorem}

\begin{proof}
For interaction~$k$, define the shorthand $f_k(\mathbf{u})\triangleq \phi_{m_k}(\mathbf{u},\gamma_k)$.
By second-order Taylor expansion with mean-value remainder, the change in interaction loss is
\begin{equation}
\Delta f_k = \nabla f_k(\bm{\vartheta}_{m_k})^{\!\top}\Delta\bm{\vartheta}_{m_k}
+ \tfrac{1}{2}\,\Delta\bm{\vartheta}_{m_k}^{\top}\nabla^{2}f_k(\widetilde{\bm{\vartheta}}_{m_k})\,\Delta\bm{\vartheta}_{m_k},
\label{eq:taylor_interaction}
\end{equation}
where $\widetilde{\bm{\vartheta}}_{m_k}$ lies on the segment between $\bm{\vartheta}_{m_k}$ and $\bm{\vartheta}_{m_k}+\Delta\bm{\vartheta}_{m_k}$.
Applying the Cauchy--Schwarz inequality to the first term and the spectral-norm bound to the quadratic form yields
\begin{equation}
|\Delta f_k| \leq S_{m_k}(\gamma_k)\,\rho_{m_k}
+ \tfrac{1}{2}\,H_{m_k}(\gamma_k)\,\rho_{m_k}^{2}.
\label{eq:per_interaction_bound}
\end{equation}
Because the path model~\eqref{eq:path_model} is a sum of per-interaction terms, the triangle inequality gives $|\Delta P_{\ell}^{\mathrm{mat}}| \leq \sum_{k=1}^{n_\ell} |\Delta f_k|$, which together with~\eqref{eq:per_interaction_bound} yields the bound~\eqref{eq:bound}.
\end{proof}

Defining $\eta_{\max}\triangleq \sup_{m,\gamma}(S_m(\gamma)\rho_m+\frac{1}{2}H_m(\gamma)\rho_m^2)$, every $n_\ell$-bounce path satisfies $|\Delta P_\ell^{\mathrm{mat}}|\leq n_\ell\,\eta_{\max}$; any measured bias exceeding $n_\ell\,\eta_{\max}$ cannot be explained predominantly by material-parameter error.
Theorem~\ref{thm:material_bound} applies only to the material-dependent correction of a fixed traced path; path creation, blockage, rerouting, and interaction-order changes remain part of $e_{\ell}^{\mathrm{geo}}$ in~\eqref{eq:error}.

Fig.~\ref{fig:material_bound}(a) evaluates the bound numerically for four ITU dielectric materials by finite-difference perturbation of the Fresnel-response loss under $\pm 30\%$ relative-permittivity variation, corresponding to $\rho_m=0.30\,\varepsilon_{r,m}^{\mathrm{nom}}$.
Over the full incidence range $1^{\circ}$--$89^{\circ}$, the per-material maximum shift is bounded by 1.91~dB (concrete), 1.68~dB (glass), 3.74~dB (plywood), and 6.28~dB (wood); the global $\eta_{\max}=6.28$~dB is attained on wood.

Fig.~\ref{fig:material_bound}(b) shows the measured mean bias magnitude on the detailed InF 3D model by bounce count: 19.2, 23.5, and 23.1~dB for single-, double-, and triple-bounce MPCs.
By Theorem~\ref{thm:material_bound}, a three-bounce path has a material-correctable ceiling of $3\eta_{\max}<18.84$~dB even at $\pm 30\%$ perturbation (using the worst-case wood material), whereas the observed three-bounce bias is 23.1~dB.
Material-domain calibration alone cannot close this gap in either the InF or InH dataset; the InH uncalibrated RMSE of 23.12~dB likewise exceeds the same ceiling.

{The bound constitutes an exclusion argument: material uncertainty is ruled out as the dominant error source, so the residual predominantly originates from 3D scene inaccuracies that alter path topology, consistent with prior findings that geometry fidelity dominates RT validation error at upper mid-band~\cite{ying2026npj}.}
The remaining error is dominated by $e_{\ell}^{\mathrm{geo}}$, the path-topology mismatch caused by 3D scene inaccuracies.
Material parameters affect only the Fresnel loss along a fixed path, whereas path-level features such as bounce count, RT predicted power, TX--RX distance, and LoS status encode whether a traced path is realistic or an artifact of scene simplification.
A correction function of path-level features can therefore capture the structured component of $e_{\ell}^{\mathrm{geo}}$ without re-running RT.
{A path with low RT predicted power and high bounce count traverses more scene surfaces, accumulating mismatch from each unmodeled or displaced object; the linear correction scales with the cumulative geometry-error exposure. The LoS flag and TX--RX distance further separate direct paths from multi-bounce paths that are more susceptible to scene simplification.}

\begin{figure}[t]
\centering
\includegraphics[width=\columnwidth]{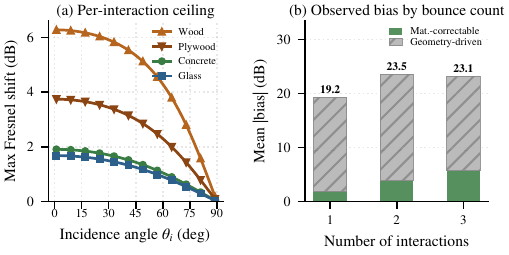}
\vspace{-25pt}
\caption{Material-sensitivity bound (Theorem~\ref{thm:material_bound}) versus observed per-path bias on the detailed InF model at 6.75 and 16.95~GHz. (a)~Maximum Fresnel shift under $\pm 30\%$ relative-permittivity perturbation; structural materials shift by under 2~dB. (b)~Mean RT--measurement error by interaction count; green and hatched bars show the material-correctable ceiling and remaining geometry-driven error.}
\label{fig:material_bound}
\end{figure}

\subsection{SPARC Calibrator}
\label{sec:method_calibrator}

The correction function in~\eqref{eq:problem} is restricted to the linear class $g(\mathbf{x})=\mathbf{w}^{\top}\mathbf{z}+b$, where $\mathbf{z}_{\ell} = \mathbf{D}_{\mathcal{S}}^{-1}(\mathbf{x}_{\ell,\mathcal{S}}-\boldsymbol{\mu}_{\mathcal{S}})$ is the fold-wise standardized feature vector restricted to a selected subset $\mathcal{S}$, in which $\boldsymbol{\mu}_{\mathcal{S}}$ and $\mathbf{D}_{\mathcal{S}}$ are the vector of feature means and the diagonal matrix of feature standard deviations of the training fold, restricted to $\mathcal{S}$, yielding the calibrated power
\begin{equation}
\widehat{P}_{\ell}^{\mathrm{cal}} = P_{\ell}^{\mathrm{RT}} - (\mathbf{w}^{\top}\mathbf{z}_{\ell} + b).
\label{eq:correction}
\end{equation}
The weights are estimated by ridge regression~\cite{hoerl1970ridge}:
\begin{equation}
\min_{\mathbf{w},b}\; \frac{1}{N}\sum_{\ell=1}^{N}\left(e_{\ell}-\mathbf{w}^{\top}\mathbf{z}_{\ell}-b\right)^2 + \alpha \lVert \mathbf{w} \rVert_2^2 .
\label{eq:ridge}
\end{equation}

Nineteen candidate features $\mathcal{F}$ are constructed from each matched path, including bounce count (and squared), three material indicators, excess delay, TX--RX distance, LoS and frequency flags, cluster size, RT predicted power, log delay, four interaction terms (bounce$\times$frequency, bounce$\times$LoS, bounce$\times$cluster, bounce$\times$zenith), and angular features derived from the departure zenith angle~$\theta_t$, arrival zenith angle~$\theta_r$, and azimuth difference~$\phi_r{-}\phi_t$ between arrival and departure directions.
Features are ranked by $|w_j|$ from a full ridge fit on each training fold; the top-$k$ subsets $\mathcal{S}_{k}=\{\pi_{1},\ldots,\pi_{k}\}$, $k=1,\ldots,K$ ($K{=}10$), are evaluated by inner leave-one-group-out CV to select
\begin{equation}
(\widehat{k},\widehat{\alpha}) =
\arg\min_{k\in\{1,\ldots,K\},\alpha\in\mathcal{A}}
\widehat{\mathrm{RMSE}}_{\mathrm{inner}}(\mathcal{S}_{k},\alpha),
\label{eq:model_selection}
\end{equation}
where $\mathcal{A}=\{0.1,1,10,50,100,500,1000\}$ is the ridge grid.
Algorithm~\ref{alg:sparse_calibration} summarizes the nested selection and deployment procedure.

\begin{algorithm}[!htb]
\caption{SPARC: Sparse Path-Aware Residual Calibration}
\label{alg:sparse_calibration}
\footnotesize
\begin{algorithmic}[1]
\STATE \textbf{Input:} Power-gated training set $\mathcal{D}{=}\{(\mathbf{x}_{\ell},e_{\ell},q_{\ell})\}_{\ell=1}^{N}$ from~Sec.~\ref{sec:system}; max feature count $K$; ridge penalty grid $\mathcal{A}$
\STATE \textbf{Output:} Calibration model $(\widehat{\mathcal{S}},\widehat{\mathbf{w}},\widehat{b},\widehat{\boldsymbol{\mu}},\widehat{\mathbf{D}})$
\STATE \textit{\% Phase 1: Nested leave-one-group-out cross-validation}
\STATE Partition $\mathcal{D}$ into $O$ folds by group label $q_{\ell}$
\FOR{$o = 1$ \TO $O$}
    \STATE $(\boldsymbol{\mu}^{(o)},\mathbf{D}^{(o)}) \leftarrow$ mean and std of training fold
    \STATE $\mathbf{z}_{\ell} \leftarrow \mathbf{D}^{(o)^{-1}}(\mathbf{x}_{\ell} - \boldsymbol{\mu}^{(o)})$ \COMMENT{standardize}
    \STATE $\mathbf{w}^{\mathrm{full}} \leftarrow$ ridge on all features \COMMENT{initial ranking}
    \STATE $\pi \leftarrow \mathrm{argsort}(|w_j^{\mathrm{full}}|,\;\text{descending})$
    \FOR{$k = 1$ \TO $K$}
        \STATE $\mathcal{S}_k \leftarrow \{\pi_1,\ldots,\pi_k\}$
        \FOR{$\alpha \in \mathcal{A}$}
            \STATE Evaluate inner cross-validation RMSE with $(\mathcal{S}_k, \alpha)$
        \ENDFOR
    \ENDFOR
    \STATE $(k^{*},\alpha^{*}) \leftarrow \arg\min$ inner cross-validation RMSE
    \STATE Refit ridge with $(\mathcal{S}_{k^*}, \alpha^{*})$; predict fold $o$
\ENDFOR
\STATE \textit{\% Phase 2: Final model}
\STATE $\widehat{\mathcal{S}} \leftarrow$ majority-vote feature subset across $O$ folds
\STATE $(\widehat{\mathbf{w}},\widehat{b},\widehat{\boldsymbol{\mu}},\widehat{\mathbf{D}}) \leftarrow$ refit ridge on all $N$ samples with $\widehat{\mathcal{S}}$
\STATE \textit{\% Deployment (no measurements, no gate)}
\STATE Run RT; for every RT path, extract $\mathbf{x}_{\ell}$ and apply~\eqref{eq:correction} $\Rightarrow \widehat{P}_{\ell}^{\mathrm{cal}}$
\end{algorithmic}
\end{algorithm}

\section{Experimental Setup}
\label{sec:setup}

RT simulations are generated with NYURay~\cite{kanhere2024nyuray,ying2026npj} using 200{,}000 launched ray samples per link, maximum depth~3, diffraction and refraction enabled, and diffuse scattering disabled.
Although the experiments use NYURay, SPARC relies only on standard RT output (predicted path power, delay, bounce count, material labels, angles) available from any commercial or open-source ray tracer such as Wireless InSite~\cite{remcom_winsite} or Sionna~RT~\cite{hoydis2023sionna}.
Dominant RT paths are matched to measured peaks across all available TX--RX--frequency groups in both environments.
RMSE in dB is the primary metric.
Leave-one-group-out CV holds out all matched MPCs from one TX--RX pair at one carrier frequency per outer fold; feature ranking, subset size, and regularization are selected on the remaining folds. The protocol measures prediction accuracy on unseen link-frequency combinations.
The matched-path observable is integrated energy of each dominant peak on a 1~GHz band-limited PDP.
Training-efficiency curves use 100 random train/test splits with the TX--RX pair as the sampling unit to quantify the number of measured links needed before RMSE saturates. Cross-environment transfer trains on all matched paths from one environment (e.g., InH) and predicts the other (e.g., InF) without adaptation, testing whether residual structure is environment-specific or universal. Plywood (InF) and wood (InH) are merged into one indicator because both materials have similar permittivity ($\varepsilon_r \approx 2$--$3$).

Five methods are compared: (i)~uncalibrated RT with ITU-R~P.2040-3 defaults~\cite{itu_p2040}; (ii)~NYURay-form material LS~\cite{kanhere2024nyuray}, fitting per-material reflection and penetration offsets via LOPO CV; (iii)~global mean offset; (iv)~gradient boosting~\cite{chen2016xgboost} (100 trees, depth~3, learning rate 0.05)~\cite{vasudevan2024mlsurvey}; and (v)~SPARC from Sec.~\ref{sec:method}.
Methods~(iv) and~(v) use the same power-gated matched MPC set, 19 features, and leave-one-group-out CV.

\section{Results and Discussion}
\label{sec:results}

\subsection{Main Calibration Comparison}

Table~\ref{tab:all_results} compares the calibration methods against field measurements. With power-gated matching, uncalibrated RT yields 18.74/23.12~dB RMSE in InF/InH; NYURay-form material LS~\cite{kanhere2024nyuray} yields 11.83/16.41~dB, consistent with Theorem~\ref{thm:material_bound}; the global offset yields 9.52/13.23~dB; and gradient boosting~\cite{chen2016xgboost} yields 5.92/5.86~dB. SPARC achieves 4.74~dB in InF and 5.39~dB in InH, outperforming gradient boosting by 1.18/0.47~dB with four linear weights versus 100 trees.
Fig.~\ref{fig:pdp_calibration} illustrates a held-out pair where uncalibrated RT underestimates multi-bounce MPCs by 15--19~dB, whereas calibrated powers closely track measurements.
At the link level, summing calibrated MPC powers per TX--RX pair reduces aggregate received-power RMSE from 9.66 to 3.41~dB across all InF groups.
Fig.~\ref{fig:method_comparison} shows the per-MPC RMSE ladder and per-pair cumulative distribution function (CDF).

\begin{figure}[!t]
    \centering
    \includegraphics[width=0.9\columnwidth]{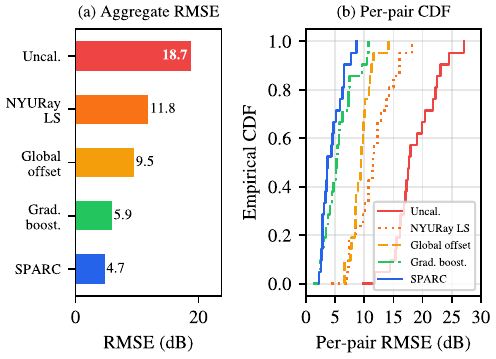}
    \vspace{-8pt}
    \caption{InF per-path calibration at 6.75 and 16.95~GHz with power-gated matching and a 1~GHz integrated-energy target. (a)~Aggregate RMSE; (b)~empirical CDF of per-pair RMSE.}
    \label{fig:method_comparison}
    \vspace{-10pt}
    \end{figure}

\begin{table}[t]
\centering
\caption{Per-path power RMSE (dB) for five methods on InF and InH at 6.75 and 16.95~GHz.}
\label{tab:all_results}
\footnotesize
\setlength{\tabcolsep}{5pt}
\begin{tabular}{@{}lcc@{}}
\toprule
Method & InF & InH \\
\midrule
Uncalibrated~\cite{itu_p2040} & 18.74 & 23.12 \\
NYURay-form material LS~\cite{kanhere2024nyuray} & 11.83 & 16.41 \\
Global mean offset & 9.52 & 13.23 \\
Gradient boosting~\cite{chen2016xgboost} (group CV) & 5.92 & 5.86 \\
SPARC (group CV) & \textbf{4.74} & \textbf{5.39} \\
\bottomrule
\end{tabular}
\par\vspace{3pt}{\footnotesize\raggedright\emph{Note:} All methods use the 1~GHz integrated-energy target and leave-one-group-out CV.\par}
\vspace{-8pt}
\end{table}

\begin{figure}[!t]
    \centering
    \includegraphics[width=0.92\columnwidth]{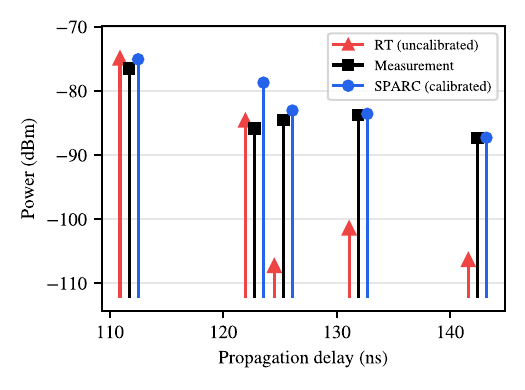}
    \vspace{-18pt}
    \caption{Matched dominant MPCs on a 1~GHz band-limited PDP for held-out InF pair TX2--RX1 at 6.75~GHz. Markers show measured, uncalibrated RT, and SPARC-calibrated powers; calibrated values are leave-one-group-out predictions. [T--R separation: 19.8~m, NLoS.]}
    \label{fig:pdp_calibration}
    \vspace{-2pt}
    \end{figure}

    \begin{figure}[!t]
    \centering
    \includegraphics[width=0.92\columnwidth]{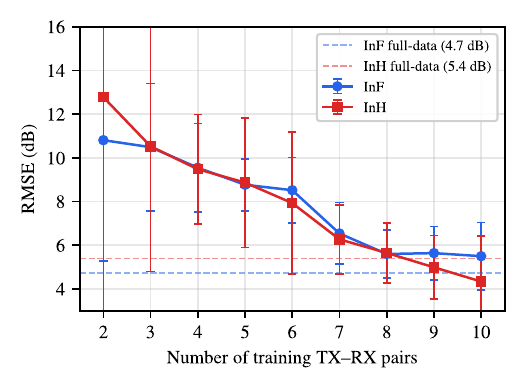}
    \vspace{-18pt}
    \caption{SPARC training efficiency on InF and InH at 6.75 and 16.95~GHz: mean per-path RMSE versus training TX--RX pairs over 100 random splits; error bars are $\pm$1 standard deviation.}
    \label{fig:training_efficiency}
    \vspace{-2pt}
    \end{figure}

\subsection{Feature Interpretation}

Nested CV retains 9--10 of 19 candidates per fold in InF and 6--8 in InH; four recur in both. The most influential feature is \emph{RT predicted power}: weaker paths carry larger errors because multi-bounce paths accumulate Fresnel uncertainty and geometry mismatch. A \emph{frequency indicator} captures the 6.75/16.95~GHz offset caused by frequency-dependent material and diffraction loss. \emph{TX--RX distance} and \emph{squared bounce count} encode propagation geometry and per-bounce accumulation. All four are standard RT outputs and require no object-level geometry or environment-specific labels, unlike~\cite{haniz2025multipath} which corrects path powers based on scatterer size.

Five linear estimators (ridge, Lasso, elastic net, Bayesian ridge, Huber) span only 0.13~dB RMSE, confirming that the choice of linear regularizer is not critical.
The all-feature ridge brings no practical benefit: under the same leave-one-group-out protocol with inner-fold regularization selection, it yields 4.81~dB per-path RMSE in InF versus 4.74~dB for the deployed sparse model (9--10 InF features); inner cross-validation never retains all 19, so sparsity aids interpretability and deployment.
The linear advantage over gradient boosting aligns with recent findings that ensemble methods degrade under extrapolation to unseen geometries~\cite{celades2026shadow}.

\subsection{Cross-Environment Generalization}

Table~\ref{tab:cross_env} summarizes generalization across the two measured indoor environments.

\begin{table}[H]
\centering
\caption{SPARC cross-environment validation at 6.75 and 16.95~GHz: per-path RMSE (dB).}
\label{tab:cross_env}
\footnotesize
\setlength{\tabcolsep}{4pt}
\begin{tabular}{@{}lc@{}}
\toprule
Protocol & RMSE (dB) \\
\midrule
Self-calibrated InF (LOPOCV) & 4.74 \\
Self-calibrated InH (LOPOCV) & 5.39 \\
Train InH, test InF & 8.80 \\
Train InF, test InH & 9.59 \\
Joint InF+InH (LOPOCV) & 5.73 \\
\midrule
Leave-one-TX-out InF (3 TX folds) & 4.99 \\
Leave-one-TX-out InH (4 TX folds) & 5.85 \\
\bottomrule
\end{tabular}
\par\vspace{1pt}{\footnotesize\raggedright\emph{Note:} Direct-transfer rows train on one environment and predict the other. Joint row pools both.\par}
\vspace{-3pt}
\end{table}

Per-environment SPARC achieves 4.74~dB (InF) and 5.39~dB (InH). Joint InF+InH training yields 5.73~dB overall (5.48/6.08~dB in InF/InH), within 1~dB of per-environment performance. Direct transfer yields 8.80/9.59~dB (InH$\to$InF/InF$\to$InH); joint training recovers most of the gap, indicating environment-specific residual structure. Leave-one-TX-out yields 4.99/5.85~dB (InF/InH). Fig.~\ref{fig:training_efficiency} shows that 7--8 training pairs achieve within 1~dB of full-data performance in both environments.

\section{Conclusion}
\label{sec:conclusion}

SPARC provides sparse per-path calibration for upper mid-band RT. The material-sensitivity bound shows that $\pm$30\% material perturbations shift each structural interaction by under 2~dB, whereas measured per-bounce bias exceeds 19~dB, identifying 3D scene inaccuracies as the dominant RT error source.
With a sparse feature subset from one completed RT run, SPARC reduces per-path RMSE from 18.74 to 4.74~dB (InF) and from 23.12 to 5.39~dB (InH), surpassing gradient boosting (100 trees) by 1.18 and 0.47~dB. Joint training yields 5.73~dB, and leave-one-TX-out validation yields 4.99/5.85~dB (InF/InH). All retained features are standard RT outputs, so no ray-tracer modification is required.

\bibliographystyle{IEEEtran_no_dash}
\bibliography{globecom_inf_calibration}

\end{document}